\documentclass[12pt]{article}
\usepackage{fullpage,url,amsmath,amsfonts,amssymb, amsthm,mathtools}
\newcommand{\transpose}{\text{${}^{\text{T}}$}}

\newcommand{\T}{\transpose}
\newcommand{\C}{\mathcal{C}}

\newcommand{\ti}{\times}

\newcommand{\al}{\alpha}

\newcommand{\de}{\delta}

\newcommand{\f}{\frac}

\newenvironment{smatrix}{\bigl(\begin{smallmatrix}}{\end{smallmatrix}\bigr)}
\newenvironment{ssmatrix}{\left(\begin{smallmatrix}}{\end{smallmatrix}\right)}

\newtheorem{thm}{Theorem}[section]
\newtheorem{proposition}[thm]{Proposition}

\newtheorem{corollary}[thm]{Corollary}

\newtheorem{lemma}[thm]{Lemma}

\newtheorem{algorithm}{\bf{Algorithm}}
\newtheorem{example}[thm]{Example}
\newtheorem{protoexample}[thm]{Prototype Example}

\def\rank{\operatorname{rank}}
\newcommand{\Z}{\mathbb{Z}}

\newcommand{\F}{\mathcal{F}}

\title{On codes induced from Hadamard matrices \footnote{ AMS Classification 2020: 94B05, 94B10, 15B99}}
\author{Ted Hurley\footnote{University of Galway (previously National University of Ireland Galway) \\ email: Ted.Hurley@universityofgalway.ie}}
\date{}

\begin{document}
\maketitle 

\begin{abstract}

 Unit derived schemes applied to Hadamard matrices are used to construct and analyse  linear block and convolutional codes. Codes are constructed  to prescribed types, lengths and rates and multiple series of  self-dual, dual-containing, linear complementary dual and quantum error-correcting of both linear block {\em and} convolutional codes are derived. 
  \end{abstract}
\section{Introduction} 
Unit schemes form a basis for the  algebraic construction and analysis of  linear block and convolutional codes and these are  described in \cite{hurleyultimate} and references therein. 
The  non-existence of general algebraic methods for constructing, designing and studying convolutional codes has often been a problem and limited very much their size and availability, see for example McEliece \cite{mceliece} and also \cite{alm1,alm2,guar,mun}.
Here methods derived in \cite{hurleyultimate} are extended for use  on Hadamard matrices to provide constructions of linear block and convolutional codes and to construct these to required types, distances and rates.  The work here can be read independently of \cite{hurleyultimate} although the ideas initiated in \cite{hurleyultimate} are in the background.
The types constructed  include {\em self-dual}, {\em dual-containing}, {\em linear complementary dual} and {\em quantum} codes and large lengths, rates and distances are achievable.
The codes are given over finite fields  and types of code required 
are constructed  in both the
linear block and convolutional cases.
Methods using orthogonal units, Fourier/Vandermonde units, group ring units and related units for constructing and analysing such codes is devised  in \cite{hurleyultimate}. The methods are applied to Hadamard matrices to construct algebraically the linear block and convolutional codes and  properties of Hadamard matrices allows these to be constructed to required length, rate and type. Infinite series are derived. 
 From a single Hadamard matrix, multiple linear block and convolutional codes are formed and formed to required types. 

 The distances achieved can often be calculated algebraically; for example the distance of a rate $\frac{1}{2}$ convolutional code obtained is of the order of twice the distance of the linear block rate $\frac{1}{2}$ codes obtained from the same Hadamard matrix.

 $C^{\perp}$ denotes the dual of the code $\C$. See section \ref{note} for precise definition of dual of a convolutional code.
$\C$ is a {\em dual-containing} (DC) code if $\C \cap C^{\perp}=\C^{\perp}$;
$\C$ is a {\em linear complementary dual} (LCD) code if $\C \cap C^{\perp}=0$.
 $\C$ is a {\em self-dual code} provided $\C^{\perp}= \C$ and is an important type of dual-containing code. 
Dual-containing codes, which include self-dual codes, can be used to construct quantum error-correcting codes (QECC) by the CSS method \cite{calderbank,calderbank1,steane}; here  convolutional quantum error-correcting codes are constructed in this way from Hadamard matrices. 

A Hadamard 
matrix is an $n\ti n$ matrix $H$ with entries $\pm 1$ satisfying $HH\T= nI_n$. Such a matrix can only exist for $n=2$ or $n=4m$ for a positive integer $m$,  \cite{lint} Theorem 18.1; see \cite{horadam} for a beautifully written book on Hadamard matrices.
Here these $\pm 1$ entries may be considered as elements in a general field.

The {\em Walsh-Hadamard codes} can be formed from Walsh-Hadamard matrices of size $2^k\ti 2^k$ by equating  the entries $(-1)$ to $0$ and then forming binary codes from the $(k+1)$ linearly independent rows remaining. 
 Using the unit-derived methods on a general Hadamard $n\ti n$ matrix gives much more scope, arbitrary rates, good distances, required types but also both linear block and convolutional codes are  formed. Several algorithms exist for decoding convolutional codes, the most common ones being the Viterbi algorithm and the sequential decoding algorithm.

 Propositions \ref{lcd}, \ref{self}, \ref{same}, \ref{diff}, \ref{diff1} on codes  derived from Hadamard are proven  and these form a basis for specific algorithms.  The  following general algorithms are noted:
 \begin{itemize} \item Algorithm \ref{alg1} constructs LCD rate $\frac{r}{n}$, for $r$, $0<r<n$, linear block codes from Hadamard $n \ti n$ matrices. \item Algorithm \ref{alg2} constructs self-dual length $2n$ codes from Hadamard $n\ti n$ matrices. \item Algorithm \ref{alg3} constructs self-dual length $n$ convolutional codes from Hadamard $n\ti n$ matrices. \item Algorithm \ref{alg4} constructs dual-containing, length $n$, rate $\frac{r}{n}$, ($n >r\geq \frac{n}{2}$), convolutional codes from an $n\ti n$ Hadamard matrices.
 \end{itemize}

    The codes are  readily implemented once an expression for the Hadamard matrix is available. Large lengths and rates are obtainable. The brilliant Computer Algebra system GAP with included packages Guava and Gauss, \cite{gap}, proves extremely useful in manipulating submatrices, working over finite field, constructing applications and computing and verifying  distances. 
 
 Higher memory convolutional codes may also be generated by breaking the Hadamard matrices further into blocks 
 ; see \cite{hurleyultimate,hurleyorder} for this.
 The general process for constructing  higher memory convolutional codes from  Hadamard matrices is left for later development; however an example, \ref{higher}, is given  for a small case to show how the process  can proceed for Hadamard matrices.
 
 The notation and parameters used in the following applications may be found in section \ref{note}. 
 The notation for linear block codes is standard; however the notation and parameters used for convolutional codes  vary in the literature and the specifics used need to be clarified.

 \paragraph{Explicit samples of applications}\label{appintro}
 Applications are obtained by applying the Propositions and Algorithms to particular cases. Once the Hadamard matrix is formed, the Propositions and Algorithms may then be applied to produce multiple cases of required codes. 

  $H(n)$ denotes a Hadamard matrix of size $n$.
  

From $H(20)$ the following are formed: \begin{enumerate} \item LCD $[20,13,4]_3$, $[20,7,6]_3$ codes; \item  self-dual convolutional $(20,10,10;1,12)_{3^2}$ codes; \item\label{ju} DC convolutional $(20,13,7;1,8)_{3^2}$ codes;\item quantum codes of length $20$, distance $8$ and rate $\frac{6}{20}$ over $GF(3^2)$; \item self-dual $[20,10,8]_5$ codes;\item $(20,10,10;1,14)_{7^2}$ self-dual convolutional codes; \item $[40,20,12]_3$ self-dual codes given directly in systematic form, see Proposition \ref{self}. \end{enumerate}
 From $H(28)$ the following are formed: \begin{itemize} \item 
   LCD $[28, 16,6]_3, [28,12,9]_5$  codes;  \item Convolutional self-dual $(28, 14, 14;1,12)_{3}$ codes over $GF(3)$; \item DC convolutional  $(28,18,10;1,8)$ codes over $GF(3)$; \item quantum codes of length $28$, distance $8$, rate $\frac{8}{28}$ over $GF(3)$; \item  self-dual convolutional $(28,14,14;1, 16)$ codes over $GF(5)$ \item DC convolutional $(28,16,12;1,14)$ codes over $GF(5)$; \item $[28,14,9]_7$ self-dual codes.
 \end{itemize} 

Generally from $H(n)$ with $p\nmid n$, self-dual $(n, \frac{n}{2}, \frac{n}{2};1,d)$ convolutional  codes and DC $(n, r, n-r;1;d)$, $r > \frac{n}{2}$, convolutional codes are formed.  
In prototype example \ref{pro12}, it is shown how the different types of LCD, self-dual, DC, quantum, linear block and convolutional codes may be derived for a  small Hadamard matrix case.
  Self-dual codes over $GF(p)$ can often be  obtained from a Hadamard matrix of size  $n$ when $p\mid n, p\neq 2$, see Proposition \ref{dual1}. For example self-dual $[12k,6k,d]_3$ codes are produced from Paley-Hadamard matrices of size $12k$, see section \ref{ternary}.

An understanding of the Propositions and Algorithms  allows one to  take a Hadamard matrix and construct LCD, self-dual, DC and QECC codes therefrom. Section \ref{ternary} is given over to considering ternary codes and codes over fields of characteristic $3$. Codes over $GF(5)$ from Hadamard matrices may similarly be worked on. Using non-separable Hadamard matrices  in applications seems to work out better. 
   
  \subsection{Additional notation and background}\label{note}
 The notation  for linear block codes is standard and may be found in \cite{blahut,joh,mceliece,macslo} and many others. $GF(q)$ denotes the finite field with $q$ elements and $\Z_m$ denotes the integers modulo $m$; in particular $\Z_p = GF(p)$ for a prime $p$. A $[n,r,d]$ code denotes a linear block code of length $n$, dimension $r$, and (minimum) distance $d$; the {\em rate} is $\f{r}{n}$. A $[n,r,d]_q$ code denotes a linear block code of length $n$, dimension $r$ and distance $d$ over the field $GF(q)$.  

Different equivalent definitions for convolutional codes are given in the literature. The notation and definitions used here follow that given in \cite{ros,smar,rosenthal1}. 
A rate $\frac{k}{n}$ convolutional code with parameters $(n,k,\de)$ over a field $\F$ is 
 a submodule of $\F[z]^n$ generated by a reduced
basic matrix $G[z] =(g_{ij}) \in \F[z]^{r\ti n}$ of rank $r$ where $n$ is
the length,  $\de = \sum_{i=1}^r \de_i$ is the {\em degree}
with  $\de_i= \max_{1\leq j\leq r}{\deg g_{ij}}$. Then 
$\mu=\max_{1\leq i\leq r}{\de_i}$ is known as the {\em memory} of the
code and the code is  then given with parameters $(n,k,\de;\mu)$. 
The parameters $(n, k,\delta;\mu, d_f)$ are  used for such a code 
with free (minimum) distance $d_f$. Further $(n, k,\delta;\mu, d_f)_q$ is used to specify that the code is over the field $GF(q)$.  

 Suppose $\C$ is a convolutional code in $\F[z]^n$ of rank $k$. A generating matrix $G[z] \in
\F[z]_{k\times n}$ of $\C$ having rank $k$ is called a
{\em generator} or {\em encoder matrix}  of $\C$. 
A matrix $H \in \F[z]_{n\times(n-k)}$ satisfying $\C = \ker H =
\{v \in \F[z]^n : vH = 0 \}$ is said to be a {\em control matrix} or
    {\em check matrix} of the code $\C$. 

 Convolutional codes can be {\em
   catastrophic or non-catastrophic}; see
- for example \cite{mceliece} for the basic definitions. A
catastrophic convolutional code is prone to catastrophic error
propagation and is not much use. A
convolutional code described by a generator matrix  with {\em right
polynomial inverse} is a non-catastrophic code; this is sufficient for
our purposes. The designs given here for  the generator matrices allow for
specifying directly the control matrices and the right polynomial
inverses. 
Lack of algebraic construction methods for designing convolutional codes limited their size
and availability, see McEliece \cite{mceliece} for discussion and also \cite{alm1,alm2,guar,mun}. It is shown here how Hadamard matrices can be used to construct convolutional codes but see also \cite{hurleyultimate,required}. 
Several algorithms exist for decoding convolutional codes, the most common ones being the Viterbi algorithm and the sequential decoding algorithm.

Let $G(z)$ be the generator matrix for a convolutional code $\C$ with memory $m$.
Suppose $G(z)H\T(z) = 0$, so that $H\T(z)$ is a control matrix, and then $H(z^{-1})z^m$  generates the {\em convolutional dual code} of $\C$, see \cite{dual2} and \cite{dualconv}. This is also known as the {\em module-theoretic dual code}.\footnote{In convolutional coding theory, the idea of {\em dual code} has 
two meanings. 
The other dual convolutional code defined  
  is called {\em the sequence space dual}; the generator
  matrices for these two types are related by a specific formula.}
The code is then dual-containing provided the code generated by $H(z^{-1})z^m$ is contained in the code generated by $G(z)$.

The dual of a code $\C$ is denoted by $\C^{\perp}$. 

Dual-containing (DC) codes, which contain self-dual codes, are an important class of codes for theoretical and practical purposes. Besides their direct applications,  DC codes are used to construct quantum error correcting codes (QECC) by the CSS method \cite{calderbank,calderbank1,steane}. 
Here then {\em  quantum error correcting} linear and convolutional codes of  different lengths and rates  are constructed explicitly from Hadamard matrices. 

 Linear complementary dual, LCD, codes have been studied extensively for their theoretical and practical  importance's  by
Carlet, Mesnager, Tang, Qi and Pelikaan, \cite{sihem3,sihem2,sihem4, sihem},  and were originally introduced by Massey in \cite{massey,massey2}.  They have been
used  for improving the security of information on sensitive devices
against {\em side-channel attacks} (SCA) and {\em fault non-invasive
  attacks}, see \cite{carlet}, and  in {\em data
  storage} and {\em communications' systems}.  
Here  LCD linear block and convolutional codes are constructed from Hadamard matrices using the unit-derived and associated methods.

 The {\em unit-derived} and associated methods for constructing and analysing linear block codes  was initiated in \cite{hur0, hur1,hur2,hur11} and developed further in \cite{unitderived,hurleyquantum}.
The papers \cite{hurleyultimate,required} and references therein extend the unit-derived and related methods ideas to form in addition convolutional codes and algebraic methods for constructing whole series of linear block and convolutional codes to prescribed length, distance, rate and type are derived. 
The unit-derived methods give further information on the code in addition to describing the  generator and control matrices. 
McEliece - see for example \cite{mceliece} - remarks: `A most striking fact is the lack of algebraic constructions of families of convolutional codes'; constructing convolutional codes of reasonable length was beyond computer generation. 

With which Hadamard matrix of a particular size $n$ should one work? It is not an issue in cases where just the main properties of a Hadamard matrix are required. 
It seems best from practice and intuition  to work with non-separable Hadamard matrices. Here {\em non-separable} means the Hadamard matrix  is not  a non-trivial tensor product of other Hadamard matrices. The Walsh-Hadamard matrices are separable except for size  $2$.

\section{Linear block and convolutional codes induced from Hadamard matrices}

The Walsh-Hadamard binary linear block codes  $[2^k,k,2^{k-1}]_2$ and $[2^k, k+1, 2^{k-1}]_2$ have very small rate   but have found use probably on account of the distances and decoding methods available. 
Codes from general Hadamard matrices as now described allow much more scope with much  better rates, good distances, required types and both linear block and convolutional codes may be formed.   

Let $H$ be a Hadamard matrix with $HH\T = nI_n$. Break $H$ as $H= \begin{ssmatrix} A \\ B \end{ssmatrix} $ for $A$ an $r\ti n$ matrix and then $ \begin{ssmatrix} A \\ B \end{ssmatrix} \begin{ssmatrix}  A\T & B\T \end{ssmatrix}=nI_n$. When $n\neq 0$ in the field under consideration, a code is obtained in which $A$ is the generator matrix and $B\T$ is a check matrix. This is the basic method for producing the linear block codes from Hadamard matrices. A more general method is to take any $r$ rows of the Hadamard matrix to generate a code and a check matrix is obtained by eliminating the corresponding columns of the transpose of the Hadamard matrix.  The convolutional codes are essentially obtained from breaking the Hadamard matrix into blocks and using the blocks as `components' of the generator matrix of a convolutional code. Properties of the Hadamard matrix are used to construct  the {\em type} of code required. 

The following Propositions \ref{lcd}, \ref{self}, \ref{same}, \ref{diff} and \ref{diff1} form the basis for deriving algorithms with which series of  codes, both linear block and convolutional, are derived. 
  These allow linear block and convolutional codes of particular types, such as self-dual, DC, LCD and quantum, to be constructed and enables these to be 
 devised to required length and rate. Algorithms \ref{alg1}, \ref{alg2}, \ref{alg3}, \ref{alg4} follow from which numerous  applications can be devised.

 It is worth noting that arithmetic over $GF(p)=\Z_p$ is simply modular arithmetic and is easily implemented. 
    Let $d(X)$ denote the distance of the linear block code generated by the matrix $X$.  
  
\begin{proposition}\label{lcd} Let $H$ be a  Hadamard matrix of size $n$ and $n\neq 0$ in a field $\F$. Suppose $H$ has the form 
      $H= \begin{ssmatrix} A \\ B \end{ssmatrix}$, where $A$ has size $r\ti n$, implying  $  \begin{ssmatrix} A \\ B \end{ssmatrix} \begin{ssmatrix}  A\T & B\T \end{ssmatrix}=nI_n$. Then the code generated by $A$  over $\F$ is an LCD $[n,r]$ code $\mathcal{A}$ and $B$ generates the dual code of $\mathcal{A}$. 
    \end{proposition}
    \begin{proof} Both $A$ and $B$ have full ranks as $H$ is invertible.  Now $AB\T=0$ and so $B\T$ is a control matrix for the code $\mathcal{A}$ and thus $B$ generates the dual code of $\mathcal{A}$. Since $H$ is invertible in $\F$ a combination of the rows of $A$ cannot be a non-trivial combination of the rows of $B$ and thus $\mathcal{A}$ is an LCD code.
      \end{proof}
    \begin{algorithm}\label{alg1} {\em Construct LCD rate $\frac{r}{n}$ linear block codes from Hadamard $n\ti n$ matrices $H$ as follows:}
      
      Let $\F= GF(p)$ where $p\nmid n$. Choose any $r$ rows of $H$ to form the generator matrix of an $[n,r]$ code over $\F$. This code is an $[n,r]$ LCD code.
      \end{algorithm}
    \begin{proposition}\label{self} Let $n\neq 0$ in a field $\F$ and $H$ a Hadamard matrix with $HH\T = nI_n$. Then there exists $\al \in \F$ or in a quadratic extension of $\F$ such that $(I_n,\al H)$ generates a self-dual code.
      
    \end{proposition}
    \begin{proof} Let $I=I_n$.
      Now $(I,\al H)\begin{ssmatrix}I \\ \al H\T \end{ssmatrix} = I +\al^2nI = (1+\al^2n)I$. Now $(1+x^2n=0)$ has a solution in $\F$ or else $(1+x^2n)$ is irreducible over  $\F$. Thus in $\F$ or in a quadratic extension of $F$ there exists an $\al$ such that $(1+\al^2 n) = 0$. Then $(I,\al H) \begin{ssmatrix}I \\ \al H\T \end{ssmatrix}= 0$ and so $K\T =\begin{ssmatrix}I \\ \al H\T \end{ssmatrix}$  of rank $n$ is a control matrix. Thus $K= (I,\al H) $ generates the dual of the code and hence the code is self-dual.
    \end{proof}

    The distance of the codes in Proposition \ref{self} may be worked out from Proposition \ref{sys} but not always easily. From the self-dual code, by the CSS construction, a quantum error-correcting code may be constructed with the same distance.
    
An example of how this Proposition \ref{self} works is given with Prototype example \ref{pro12} which uses a Hadamard matrix of size $12$.
\begin{algorithm}\label{alg2} {\em Construct  $[2n,n]$ self-dual codes using Hadamard matrices of size $n$:}

  Let $\F=GF(p)$ where $p\nmid n$ and  $A= (I_n,\al H)$ where $\al$ satisfies $(1+\al^2n=0)$ in $\F$ or in a quadratic extension of $\F$. Then code generated by $A$ is self-dual.
  \end{algorithm}
\begin{proposition}\label{same} Let $H$ be a Hadamard matrix of size $n$ and $n\neq 0$ in a field $\F$. Suppose $H$ has the form 
      $H= \begin{ssmatrix} A \\ B \end{ssmatrix}$ implying $  \begin{ssmatrix} A \\ B \end{ssmatrix} \begin{ssmatrix}  A\T & B\T \end{ssmatrix}=nI_n$ where $n=2m$ and $A$ and $B$ have size $m\ti n$. 
  Let $G(z)=A+iBz$ where $i=\sqrt{-1}$ in $\F$ or in a quadratic
extension of $F$. Then $G(z)$ generates a self-dual  convolutional (non-catastrophic) code with parameters $(2m,m,m;1,d)$ where $d=d(A) + d(B)$.
  \end{proposition} 
\begin{proof}
  Now $G(z)(iB\T + A\T z)= (A+iBz)(iB\T + A\T z) = 0+nI_mz -nI_mz+0= 0$ and so $H\T(z)=(iB\T + A\T z)$ is a control matrix. Hence $H(z^{-1})z=A+iB$ generates the dual of the code and so the code is self-dual. Also
      $(A+iBz)A\T = nI_m$ and so $(A+iBz)$ has a right polynomial inverse and thus the code generated by $G(z)$ is non-catastrophic.

      The proof of the distance is straight forward and omitted.
\end{proof}
\begin{algorithm}\label{alg3} {\em Construct self-dual convolutional codes from Hadamard matrices.}

  Let $H$ be a Hadamard matrix of size $n=2m$ and $\F= GF(p)$ where $p\nmid n$. Let $A$ consist of any $m$ rows of $H$ and $B$ consist of the other $m$ rows of $H$. Define $G(z)=A+iB$ where $i=\sqrt{-1}$ in $\F$ or in a quadratic extension of $\F$. Then $G(z)$ generates a self-dual convolutional code with parameters $(2m,m,m;1,d)$ where $d=d(A)+d(B)$.
  \end{algorithm}
\begin{proposition}\label{diff} Let $H$ be a Hadamard matrix of size $n$ so that $HH\T=nI_n$ and $n\neq 0$ in a field $\F$. Suppose $H$ has the form 
      $H= \begin{ssmatrix} A \\ B \end{ssmatrix}$ implying   $  \begin{ssmatrix} A \\ B \end{ssmatrix} \begin{ssmatrix}  A\T & B\T \end{ssmatrix}=nI_n$ where  $A$ has size $r\ti n$ and $B$ has size $(n-r)\ti n$ with $r>(n-r)$. Let $t=(2r-n)$ and define $B_1= \begin{smatrix} 0_{t\ti n} \\ B\end{smatrix}$. Then
   $G(z) = A+iB_1z$ generates a convolutional dual-containing $(n,r,n-r;1,d)$ code $\C$ where $i=\sqrt{-1}$ in $\F$ or in a quadratic extension of $\F$.

  \end{proposition}
  
\begin{proof}  Define $0_t=0_{t\ti n}$. Thus
$B_1 = \begin{ssmatrix} 0_t \\ B\end{ssmatrix}$ is an $r\ti n$ matrix. Now $A\T$ is an $n\ti r$ matrix and thus has the form $A\T=(X,C_1)$ where $C_1$ has size $n\ti (n-r)$ and $X$ has size $n\ti (2r-n)$. As $AA\T = nI_r$ then $AC_1= n\begin{pmatrix} 0_{(2r-n) \ti (n-r)} \\ I_{(n-r)\ti (n-r)}\end{pmatrix}$ and also $B_1B=n\begin{pmatrix} 0_{(2r-n) \ti (n-r)} \\ I_{(n-r)\ti (n-r)}\end{pmatrix}$.
Now    $A= \begin{ssmatrix}X\T \\ C_1\T\end{ssmatrix}, B_1=\begin{ssmatrix} 0_t \\ D\T \end{ssmatrix}$. Then  $(A+iB_1z) (iB\T+C_1z) = 0$ so $H\T(z)=(iB\T+C_1z)$ is a control matrix and $H(z^{-1})z=C_1\T+iB$ generates the dual of the code. The code generated by $C_1\T+iB$ is easily seen to be contained in the code generated by $(A+iB_1)$ and so the code generated by $(A+iB_1)$ is dual-containing. That the code is non-catastrophic follows in a similar manner to the proof in Proposition \ref{same}. 

    \end{proof}

Suppose $H=\begin{ssmatrix} P \\ Q \end{ssmatrix}$ so that
$\begin{ssmatrix} P \\ Q \end{ssmatrix}\begin{ssmatrix} P\T & Q\T \end{ssmatrix} = nI_n$ where $P$ has size $r$ with $r> \frac{n}{2}$. Then this can be written $\begin{ssmatrix} A \\ B \\C \end{ssmatrix}\begin{ssmatrix} A\T & B\T &C\T\end{ssmatrix} = nI_n$ where $C$ has the same size as $B$. Another way to look at Proposition \ref{diff} is as follows:

  \begin{proposition}\label{diff1} Let $H$ be a Hadamard $n\ti n$ matrix and $n \neq 0$ in $\F$. Suppose $H=\begin{ssmatrix} A \\ B \\ C \end{ssmatrix}$ where  $C$ has the same size as $B$ and
    thus   $\begin{ssmatrix} A \\ B \\C \end{ssmatrix}\begin{ssmatrix} A\T & B\T &C\T\end{ssmatrix} = nI_n$.

      Then $G(z)=\begin{ssmatrix} A \\ B \end{ssmatrix} + i \begin{ssmatrix}\underline{0} \\ C \end{ssmatrix}$ defines a dual-containing convolutional $(n,r,n-r;1,d)$ code where $i=\sqrt{-1}$ in $\F$ or in a quadratic extension of $\F$, $\underline{0}$ is the zero matrix of the same size as $A$ and $r\ti n$ is the size of $\begin{ssmatrix} A \\ B \end{ssmatrix}$.
 
  \end{proposition}

  This is equivalent to Proposition \ref{diff} but a proof is given as it's instructive for the algorithm that follows.

  \begin{proof}  Note that $BB\T = nI_t=CC\T$ for some $t$. Use $0$ for a zero matrix whose size is clear from the context. 
    Now $(\begin{ssmatrix} A \\ B \end{ssmatrix} + i \begin{ssmatrix} 0 \\ C \end{ssmatrix}z)(iC\T + B\T z)= 0+\begin{ssmatrix}0\\ I_t \end{ssmatrix} - \begin{ssmatrix}0\\ I_t \end{ssmatrix} + 0 = 0$ and so $H\T(z)= iC\T + B\T z$ is a control matrix for the code. Hence $H(z^{-1})z = B + iC$ generates the dual of the code. It is easy to see that the code is dual-containing. Also a right inverse for $G(z)$ is readily written down and so the code is non-catastrophic.  \end{proof}
             
This can be used to find or estimate the distances of the dual-containing codes derived. 
\begin{algorithm}\label{alg4} {\em Construct rate $\frac{r}{n}$, $n> r \geq n/2$,
  dual-containing convolutional codes from Hadamard matrices of size $n$.}

  Let $\F=GF(p)$ where $p\nmid n$, $A$ consist of $r$ rows of $H$ and $B$ consist of the other $(n-r)$ rows of $H$. Define $G(z) = A + iB_1z$ where $B_1= \begin{ssmatrix} 0_{t\ti n} \\ B \end{ssmatrix}$, $t= 2r-n$ and $i=\sqrt{-1}$ in $\F$ or in a quadratic extension of $\F$. Then $G(z)$ generates a dual-containing convolutional code $(n,r,n-r,1,d)$.
\end{algorithm}

The distance $d$ in Algorithm \ref{alg4} can be estimated from Proposition \ref{diff} as follows: Let $A_1$ be the matrix of the first $(2r-n)$ of $A$ in Proposition \ref{diff}; the distance of $\C$ is then $\min\{d(A_1), d(A) + d(\begin{ssmatrix} A_1 \\ B \end{ssmatrix})\}$.


Note that from a dual-containing code, by the CSS construction, a quantum error-correcting codes, QECC,  of the same length and distance as that of the dual-containing code is constructible.

The Hadamard matrix over $GF(3)$ has entries $\{1,-1\}$ which are all the non-zero entries of $GF(3)$. $GF(5)$ has the property that it contains a square root of $(-1)$ as $2=\sqrt{-1}$ in $GF(5)$. But also $GF(3)$ may be extended to $GF(3^2)$ which has a square root of $(-1)$; the significance of $\sqrt{-1}$ is clear from  the  convolutional codes derived as in Propositions \ref{same} and \ref{diff}.

    For characteristic dividing $n$ the rank of a Hadamard $n\ti n$ matrix is then less than $n$. When the characteristic does not divide $n$ then the Hadamard  $n\ti n$ matrix has rank $n$ and its rows are independent; this is used in Proposition \ref{diff} and Proposition \ref{same}. 
    Proposition \ref{self} uses a Hadamard matrix to give a generator matrix of a self-dual code in  {\em systematic form}, \cite{blahut}. The distance can be obtained from the Hadamard matrix as follows. 
    \begin{proposition}\label{sys}(Proposition 3.8 in \cite{hurleyultimate}.) Let $\C$  be the code generated by $G=(I_n,P)$. Suppose the code generated by any $s$ rows of $P$ has distance $ \geq (d-s)$ and for some choice of $r$ rows the code generated by these $r$ rows has distance exactly $(d-r)$, then the distance of $\C$ is $d$.    
    \end{proposition}

  For a matrix $K$ the following notation, as suggested by \cite{gap}, is adopted: $K[s..t][u..v]$ is the submatrix of $K$ consisting of the rows $s$ to $t$ of $K$  and the columns $u$ to $v$ of $K$.

  \begin{lemma}\label{rank} Let $H=H(n)$ be a Hadamard matrix of size $n$ and and let $p\neq 2$ be a prime divisor of $n$. Then $\rank(H) \leq \frac{n}{2}$ in $\Z_p=GF(p)$. \end{lemma}
        \begin{proof} Modulo $p$, $HH\T = nI_n = 0_{n\ti n}$. If $A$ is an $m\ti n$ matrix and $B$ is $n \ti k$, then $\rank(A)+\rank(B) - n \leq \rank(AB) $. Let $H= A, B=H\T$. Then
          $\rank(H)+ \rank(H\T)- n \leq \rank HH\T =0$. But $\rank (H)=\rank(H\T)$ and so $2\rank(H)\leq n$ and hence $\rank(H) \leq \frac{n}{2}$.
        \end{proof}
        It is easy to check the rank as required in such cases when $p\mid n$. In many cases it works  that the $\rank$ is actually $\frac{n}{2}$ but it's not necessary that the first $\frac{n}{2}$ rows are independent.

    \begin{proposition}\label{dual1} Let $H=H(2n)$ be a Hadamard matrix of size $2n$ and $p\mid n$, $p$ a prime, $p\neq 2$. Suppose over $GF(p)$ that $H$ has rank $n$ and  let $A$ be an $n\ti 2n$ submatrix of rank $n$. Then $A$ generates a self-dual $[2n,n]_p$ code over $GF(p)$.
    \end{proposition}
    \begin{proof} It may be assumed that $A$ consists of the first $n$ rows of $H$ as interchanging rows of a Hadamard matrix results in a Hadamard matrix of the same rank. Thus over $GF(p)$, $HH\T$ has the form $\begin{ssmatrix} A \\ B \end{ssmatrix}\begin{ssmatrix} A\T & B\T \end{ssmatrix} = 0$. Thus $AA\T = 0$. Now $\rank(A)=n$ and thus $\rank(A\T)=n$. Hence $A$ generates a $[2n,n]$ code $\mathcal{A}$ and $A$ is a control matrix of this code. Thus the dual code of $\mathcal{A}$ is generated by $A$ and so $\cal{A}$ is self-dual. 
      \end{proof}

    In Proposition \ref{dual1} any $n$ rows of $H$ which form a matrix of rank $n$ can be used to generate a self-dual code. In many cases the matrix of the first $n$ rows  has rank $n$  and also in many cases a selection of any $n$ rows has rank $n$. The distance may be found  for lengths up to about a 100 by computer and after that algebraic methods are required. 
    
    \paragraph{Applications/examples}
    
    Applications are derived  by applying the Propositions and Algorithms.
    The introduction lists some applications and the following is a further selection. 

    The first application is a small prototype example with $H=H(12)$ and this demonstrates how the different methods for constructing linear block and convolutional codes from Hadamard matrices codes can be developed. 
    
    \begin{protoexample}\label{pro12} Let  $H$ be  a  $12 \ti 12$ Hadamard matrix. This is a good example and  is the first case of a Hadamard matrix with size $n>2$ where $H$ cannot be derived as a Walsh-Hadamard matrix type.
      \begin{itemize} 
\item  Then $HH\T=I_{12}$ is broken up as: $  \begin{ssmatrix} A \\ B \end{ssmatrix} \begin{ssmatrix}  A\T & B\T \end{ssmatrix}=12I_{12}$. Let $A$ consist of the first 6 rows of $H$, $A=K[1..6][1..12]$,  and $B$ consist of the last 6 rows of $H$. Over $GF(3)$ this becomes 
    $  \begin{ssmatrix} A \\ B \end{ssmatrix} \begin{ssmatrix}  A\T & B\T \end{ssmatrix}=0$. Now $A$ (and $B$) have rank $6$ as does $A\T$ giving $AA\T=0$. Thus the code generated by $A$ has dual code generated by $A\T\T=A$ and so is self-dual. The distance of the code, $\mathcal{A}$, generated by $A$ is $6$ and so $\mathcal{A}$ is a $[12,6,6]$ self-dual code over $GF(3)$. This is best possible. 
      The code can correct up to $2$ errors and thus a combination of one or two rows of $A\T$ is unique and can be used to correct up to two errors in a straight-forward manner.

\item Over $GF(5)$ the code generated by $A$ is an LCD code as is the dual code is generated by $B$. Both are  $[12,6,6]_5$ codes over $GF(5)$.

    Over $GF(5)$ define $G(z)= A + i Bz$ where $i=\sqrt{-1}$; in this case $i=2$ as $2^2=4 = -1$ in $GF(5)$. $G(z)$ generates a convolutional memory $1$ code which is non-catastrophic as $(A+iBz)A\T = 6I_6=I_6$ so that $A+iBz$ has a right polynomial inverse. Now $(A+iBz)(iB\T+A\T z) = 0$ so $K\T(z) = iB\T+A\T z$ is a control matrix giving that $K(z^{-1})z = A+iBz$ generates the dual code code and so the code is self-dual. The free distance of the code is the sum of the distances of the codes generated by $A$ and by $B$ which is $12$, see Proposition \ref{same}. Thus a self-dual convolutional $(12,6,6;1,12)_5$ code is obtained. From this a quantum error-correcting convolutional code is obtained with length $12$ and distance $12$ over $GF(5)=\Z_5$. 

\item\label{five} Let $H$ again be a Hadamard $12 \ti 12$ matrix and let $\C$ be the $[24,12]$ code generated by $(I_{12},\al H)$ with $\al$ to be determined.
    Then $(I,\al H)\begin{ssmatrix} I \\ \al H\T \end{ssmatrix} = I+12 \al^2I=
    (1+12\al^2)I$. Require now that $(1+12 \al^2)= 0$ in a field to be decided. In this case $K\T =\begin{ssmatrix} I \\ \al H\T \end{ssmatrix}$, which has rank $12$, is a control matrix and then $K=(I,\al H)$ generates the dual code of $\C$ and so $\C$ is self-dual.

  \item  In item \ref{five}  require that $ 1+ 2\al^2 = 0$ in characteristic $5$ which requires  $2\al^2 = -1 = 4$ which requires $\al^2 =2 $. Now $x^2-2$ is irreducible over $GF(5)$ and so extend $GF(5)$ to $GF(5^2)$ which has an element $\al^2=2$. Then over this field $(I,\al H)$ generates a self-dual code. The length of the code turns out to be $8$ and thus get a $[24,12,8]$ self-dual code over $GF(5^2)$.

      In $GF(7)=\Z_7$, $\al =2$ satisfies $1+12\al^2=0$ and so $(I, 2H)$ generates a self-dual $[24,12,8]_7$ code.  

    \end{itemize}  
    \end{protoexample}
     
\begin{example}\label{higher}  Hadamard matrices can be used to construct higher memory convolutional codes. Let $H$ be a Hadamard $12 \ti 12$ matrix, $A=H[1.3][1..12], B=H[4..6][1..12], C=H[7..9][1..12], D=H[10..12][1..12]$. Then $G(z)=A+Bz+Cz^2+Dz^3$ gives a $(12,3,9;3,24)$ convolutional code. The distance is easily computed as $d(A)=6=d(B)=d(C)=d(D)$ and $d(\begin{ssmatrix}X \\ Y \end{ssmatrix} = 6$ for $X,Y$ different elements of $\{A,B,C,D\}$.
  \end{example}

    Further applications/examples are given below. 

\begin{example}  With $H=H(72)$ any $36$ rows generate a $[72,36,18]_3$  self-dual code. With $H=H[144]$ over $GF(3)$, $72$ rows of $H$ generate a $[144,72,d]_3$ code.  
  \end{example}

 \begin{example}   Let $H=Hadamard(20)$.  Let $A=H[1..10][1..20], B=H[11..20][1..20]$. Over $GF(3)$ the codes generated by both $A$ and $B$ are $[20,10,6]$ LCD codes. Using $G(z)=A+iBz$ gives by Proposition \ref{same} a convolutional  $(20,10,10;1;12)$ self-dual code over $GF(3^2)$ where $i=\sqrt{-1}$ in $GF(3^2)$.  From this a QECC convolutional code of length $20$ and distance $12$ is obtained. 

   Over $GF(5)$ a $[20,10,8]_5$ self-dual code is obtained from $H$. But also $A=H{[1..10]}{[1..20]}$ and $B=H{[11..20]}{[1..20]}$ give $[20,10,8]_5$ codes over $GF(5)$.
 \end{example}
 \begin{example} $H=H(24)$. Over $GF(3)$ this has rank $12$. But also $A=H[1..12][1..24]$ has rank $12$ over $GF(3)$ and generates a self-dual $[24,12, 9]$ code over $GF(3)$.

   Let  $B=H[11..24][1..24]$. Then both $A$ and $B$ generate $[24,12,7]_5$ codes over $GF(5)$. Let $G(z) = A+iBz$ where $i=\sqrt{-1}=2$ in $GF(5)$. Then by Proposition \ref{same},  $G(z)$ generates a convolutional self-dual $[24,12,12;1, 14)$ code over $GF(5)$. From this a QECC convolutional code of length $24$ and distance $14$ is derived over $GF(5)$.
 \end{example}

 \begin{example} Let $H=H(40), A=H[1..20][1..40], B=H[21..40][1..40]$. Over $GF(5)$ $H$ has rank $20$ but $A$ has rank $10$ over $GF(5)$. Now $C=H[1..10][1..40]$ has rank $10$ over $GF(5)$ and generates a $[40,10,16]$ LCD code over $GF(5)$.  Over $GF(3)$ define $G(z)= A +iBz$ where $i=\sqrt{-1}$ in $GF(3^2)$. Then $G(z)$ generates a self-dual convolutional $(40,20,20;1,d)$ code where $d=d(A) + d(B)$, Proposition \ref{same}. From this a QECC convolutional code of length $40$ and distance $d$ is obtained. 
 \end{example}
 \begin{example} $H=H(36), A= H[1..18][1..36], B=H[19..36][1..36]$. Over $GF(3)$ $A$ generates a self-dual $[36, 18,12]_3$ code as does $B$.
   Define $G(z)=A+iBz$ where $i=\sqrt{-1}=2$ in $GF(5)$ and then $G(z)$ generates a convolutional self-dual code $(36,18,18;1,d)$ where $d = d(A)+d(B)$ from which a QECC convolutional code of length $36$ and distance $d$ is obtained.
   \end{example}
    
        \subsection{Ternary codes from Hadamard matrices}\label{ternary} 
        
        Ternary codes, codes over $GF(3)$, have their own interest and are the next are the next obvious cases after binary; see for example \cite{pless} but many more in the literature. Arithmetic in $GF(3)=\Z_3$ is easily implemented. Some applications given previously are ternary codes. 
        The entries of a Hadamard matrix are the non-zero elements of $\Z_3=GF(3)$ and looking at unit-derived codes formed from Hadamard matrices over $GF(3)=\Z_3$ is  particularly interesting and beneficial.  

 Lemma \ref{rank} shows that when $3\mid n$ then $\rank(H)\leq \frac{n}{2}$ in $\Z_3=GF(3)$ for a Hadamard matrix of size $n$. For a Paley-Hadamard matrix $H$ of size $n$,  $\rank(H) = \frac{n}{2}$. Is it true in other cases?   In cases where the rank is $\frac{n}{2}$ a self-dual $[n,\frac{n}{2},d]_3$ code is constructible from any   $\frac{n}{2}$ independent rows of the Hadamard matrix.  

 The following is a consequence of Proposition \ref{dual1}: 
 \begin{proposition} Let $H$ be a Hadamard matrix of size $n$ such that $3\mid n$ and that $\rank(H) = \frac{n}{2}$. Then any submatrix of size $\frac{n}{2}\ti n$ of rank $\frac{n}{2}$ over $GF(3)$ generates a self-dual $[n,\frac{n}{2},d]_3$ code.
 \end{proposition}
 \begin{corollary}\label{tum} Let $H=H(12k)$ be a Hadamard matrix of rank  $6k$ over $GF(3)$. Then $6k$  linearly independent rows over $GF(3)$ of $H$  generate a self-dual $[n,\frac{n}{2},d]_3$ code.
 \end{corollary}
 It is interesting to find the distances attained. For a self-dual $[n,\frac{n}{2},d]_3$ ternary code it is known that $d \leq \lfloor\frac{n}{12}\rfloor+3$ \cite{mallow}. For $n=12k$ extremal ternary self-dual codes exist  for lengths $n=12,24,36,48, 60$ and do not exist for $n=72,96,120$ and for $n \geq 144$. Now by Corollary \ref{tum} $\frac{n}{2}$ linearly independent rows over $GF(3)$ of a Hadamard matrix of size $n=12k$ generate a self-dual ternary code; when $n=12,24,36,48, 60$ it is verified by computer that these are optimal. 
        
        \begin{lemma} In characteristic $3$ a non-zero sum of $r$ rows of a Hadamard $n\ti n$ matrix is the same as the sum of the first $r$ rows of a Hadamard $n\ti n$ matrix.
        \end{lemma}
        \begin{proof} In characteristic $3$ the non-zero coefficients in a sum of rows are $\pm 1$ only. Interchanging rows of a Hadamard matrix or multiplying any row by $-1$ results in a Hadamard matrix. Thus taking the relevant rows and placing them in the first $r$ places and multiplying the row by $-1$ if the coefficient is $-1$ results in a Hadamard matrix whose sum of the first $r$ rows is the same as the vector sum of the required rows.
        \end{proof}

        Thus if a lower bound on the support of the sum of the first $s$ rows of a particular type of Hadamard matrix over $\Z_3$ can be obtained then distances of the unit-derived codes from such a matrix $H$ are calculated and also the distances of the self-dual codes $(I,\al H)$ as in Proposition \ref{self} are obtained. 
        The following Proposition is a special case of Proposition \ref{lcd}.

\begin{proposition}        Let $H$ be a Hadamard $n\ti n $ and $3 \nmid n$. Then any $r$ rows of $H$ generates an LCD ternary $[n, r, d]_3$ code over $\Z_3$.
\end{proposition}
Application:
With $n=20$ the following LCD codes are obtained:

$[20,5,10]_3,[20,6,10]_3,[20,10,6]_3,[20,11,5]_3,[20,13,4]_3$

With $n=28$ the following LCD codes are obtained: $[28,7,12], [28,14,6], [28,18,4]$.

The following Proposition is immediate from  Propositions \ref{same},\ref{diff}.
\ref{diff1}.
\begin{proposition}\label{three} Let $H=H(n)$ be a Hadamard matrix of size $n$  where $3 \nmid n$.

  (i)  Suppose  $A$ consists of  $\frac{n}{2}$ rows of $H$ and $B$ consists of the other $\frac{n}{2}$ rows of $H$.
  Then  $G(z)= A+iBz$ with $i=\sqrt{-1}$ in $GF(3^2)$ generates a self-dual convolutional $(n,\frac{n}{2}, \frac{n}{2};1,d)$ code in $GF(3^2)$ where $d=d(A) + d(B)$.

  (ii) Suppose  $A$ consists of  $r$ rows of $H$ with $r > \frac{n}{2}$ and $B$ consists of the other $(n-r)$ rows of $H$. Define $B_1= \begin{ssmatrix} \underline{0} \\ B \end{ssmatrix}$ where $\underline{0}$ is the zero $(2n-r)\ti n$ matrix. Then $G(z) = A+iB_1z$ generates a dual-containing convolutional $(n,r,n-r;1,d)$ code in $GF(3^2)$ where $i=\sqrt{-1}$.
\end{proposition}

Proposition \ref{diff1} is used to calculate the distance $d$ in part (ii) of Proposition \ref{three}. 

\end{document}